\documentclass[10pt]{article}

\usepackage{latexsym}
\usepackage[tbtags]{amsmath}
\usepackage{amsfonts}
\usepackage{amssymb}
\usepackage{graphics}
\usepackage{theorem}
\usepackage{md2-rev2} % author-defined macros package

\begin{document}

\title{\setlength{\baselineskip}{2em}Macroscopic Determinism in Interacting Systems Using Large Deviation Theory}
\author{Brian R. La Cour\footnote{To whom correspondence should be addressed; e-mail: {blacour@arlut.utexas.edu}} \footnote{Current address: Applied Research Laboratories, The University of Texas at Austin, P.O. Box 8029, Austin, Texas 78713-8029} and  William C. Schieve\footnote{Ilya Prigogine Center for Studies in Statistical Mechanics and Complex Systems, Department of Physics, University of Texas at Austin, Austin, Texas 78712}}
\date{December 17, 2001}
\maketitle

\begin{abstract}
\setlength{\baselineskip}{2em}
We consider the quasi-deterministic behavior of systems with a large
number, $n$, of deterministically interacting constituents.  This work
extends the results of a previous paper [\textit{J. Stat.\ Phys.}\
\textbf{99}:1225--1249 (2000)] to include vector-valued observables on
interacting systems.  The approach used here, however, differs
markedly in that a level-1 large deviation principle (LDP) on joint
observables, rather than a level-2 LDP on empirical distributions, is
employed.  As before, we seek a mapping $\psi_{t}$ on the set of
(possibly vector-valued) macrostates such that, when the macrostate is
given to be $a_0$ at time zero, the macrostate at time $t$ is
$\psi_{t}(a_0)$ with a probability approaching one as $n$ tends to
infinity.  We show that such a map exists and derives from a
generalized dynamic free energy function, provided the latter is
everywhere well defined, finite, and differentiable.  We discuss some
general properties of $\psi_{t}$ relevant to issues of irreversibility
and end with an example of a simple interacting lattice, for which an
exact macroscopic solution is obtained.
\end{abstract}


%%%%%%%%%%%%%%%%%%%%%%%%%%%%%%%%%%%%%%%%%%%%%%%%%%%%%%%%%%%%%%%%%%%%%%%%%%%%

\section{Introduction}
\label{sect:intro}

Macroscopic determinism is a term we have used in a previous work \cite{LaCour2000} to describe the quasi-deterministic behavior of certain macroscopic observables on systems with many degrees of freedom.  The idea is that, given the initial macrostate of such an observable, its state at any future (or past) time is predictable with a probability approaching one as the number of degrees of freedom tends to infinity.  Thus, macroscopic determinism is equivalent to convergence in probability for each given time and so corresponds to a kind of law of large numbers.

In the natural world, macroscopic determinism is a quite familiar phenomenon.  It is manifest in hydrodynamics, as described by the Navier-Stokes equation, as well as in the laws of mass diffusion and chemical reactions, to name a few.  Traditionally, these macroscopic laws have been obtained either empirically or through approximations of the exact microscopic dynamics.  Examples of this approach include the Boltzmann and master equations, which become exact only in the weak interaction, long time (``$\lambda^2 t$'') and Grad limits \cite{Spohn,Davies,Grad}.  The early work of van Kampen and others \cite{vanKampen1961,Kurtz1972}, for example, gives a systematic method of obtaining macroscopic differential equations starting from an assumed master equation whose transition rates possess special extensive properties.
% begin CHANGE (1st revision)
In a similar approach, Onsager and Machlup in their well-known 1953 paper \cite{Onsager_and_Machlup1953} on thermodynamic fluctuations use a Langevin model to predict the probabilistically most likely trajectory of a set of macroscopic observables in the near-equilibrium regime.  Much later Eyink, using a large deviation theorem from Kipnis, Olla, and Varadhan \cite{Kipnis1989}, extended this classic result to the nonlinear regime for a stochastic exclusion process on a lattice gas, validating rigorously an earlier proposal by Graham \cite{Graham1981}.  More recently, the Onsager principle and Onsager-Machlup theory of fluctuations has been studied from the point of view of large deviations by Oono \cite{Oono1993} using a time-averaging approach.  Our approach differs markedly from these more traditional methods in that we consider only rigidly deterministic microscopic dynamics rather than an underlying stochastic process.

It should be remarked that the view of macroscopic time evolution presented here differs from that of traditional approaches which attempt to derive a differential equation, usually first order in time, for the macroscopic variables of interest.  This is the approach taken, e.g., by Mori \cite{Mori1958}, van Kampen \cite{vanKampen1961}, Emch and Sewell \cite{Emch_and_Sewell1968}, and many others.  Common to this approach is the use of a ``mesoscopic'' time scale over which the microscopic dynamics are assumed to be sufficiently complex that the macrostates obey (approximately) a Markov property, from which a first-order macroscopic law may be derived.  In our approach there is no such coarsening of time, and, furthermore, there is no reason to demand this as a necessary condition for defining macroscopic time evolution.  In addition, our approach has important implications for understanding irreversibility, for it implies that a semigroup property and monotonic approach to equilibrium are not to be expected in general, unlike the Markovian case.  In a future paper \cite{LaCour2002} we shall address more specifically the origin and validity of first order macroscopic laws by a considering a derivation of the Onsager linear relations from results of the present work.
% end CHANGE

% begin CHANGE (2nd revision)
Although we focus on time as a fixed parameter, our formalism accomidates an asymptotic rescaling from microscopic to macroscopic time scales as well.  Thus, one may consider the time-averaged macroscopic observable
\begin{equation}
\tilde{G}_{n,\tilde{t}} := \frac{1}{\tau_n} \int_{(\tilde{t}-1/2)\tau_n}^{(\tilde{t}+1/2)\tau_n} G_{n,t} \; \d t,
\end{equation}
for example, where $\tau_n\to\infty$ as $n\to\infty$.  Macroscopic determinism then corresponds to the convergence in probability of $\tilde{G}_{n,\tilde{t}}$, conditioned on an initial value of $a_0$, to some $\tilde{\psi}_{\tilde{t}}(a_0)$.  Such a description may be appropriate if a typical realization of $G_{n,t}$ is, say, of the form $\tau_n \e^{-\gamma t/\tau_n}$.  We assert, however, that a rescaling of the time axis is not always necessary to achieve macroscopic determinism.  Indeed, for the example considered here, no such rescaling is needed.
% end CHANGE

There is an unfortunate and inevitable lack of examples in which macroscopic laws are obtained from the \emph{exact} microscopic dynamics.  (See, however, \cite{Huerta1971}.)  This situation is inevitable, as by its very nature macroscopic determinism concerns systems with an intractable number of degrees of freedom.  It is unfortunate, however, as much confusion concerning deep questions of irreversibility and determinism has become entangled in the complexity of approximation schemes needed to obtain tractable solutions \cite{Sklar,Lebowitz1999}.  In an earlier work \cite{LaCour2000}, we addressed this problem with some generality for the case of bounded, extensive, and scalar-valued observables on systems of dynamically independent and identical particles.  This approach used results from the theory of large deviations \cite{Dembo_and_Zeitouni,Ellis} which were not readily transferable to the case of interacting systems.

The present work addresses the case of interacting systems using a significantly different and greatly more general approach.  The basic problem at hand consists in determining the time evolution of a time-dependent macroscopic observable $G_{n,t}$, which may in general be vector valued, when conditioned on a given value $a_0$ of the initial observable $G_{n,0}$.  (The microscopic dynamics are assumed to be time translation invariant, and $t$ may be positive or negative.)  The parameter $n$ is a macroscopic index representing, e.g., the number of particles in the system.  The approach used here is to consider the joint distribution of $(G_{n,0}, G_{n,t})$, then impose a conditioning constraint such that, as $n$ tends to infinity, $G_{n,0}$ converges in probability to $a_0$ while $G_{n,t}$ converges to some $a_t$ to be determined.  If, given $a_0$, there is a uniquely corresponding $a_t$, then this relation defines a macroscopic map, $\psi_{t}$, such that $a_0 \mapsto a_t = \psi_t(a_0)$.  In this case, it will be shown that the predicted macrostate takes the form of a canonical expectation in which the analogous temperature parameter is selected according to $a_0$.

We begin in Sec.\ \ref{sec:problem_description} with a more detailed description of the general problem being considered and an outline of the general approach used.  The main results are contained in Sec.\ \ref{sec:ldp_approach}.  Both standard and some novel techniques from the theory of large deviations are employed to obtain the predicted final macrostate, $a_t$, from the specified initial macrostate, $a_0$, and to establish convergence in probability of $G_{n,t}$ to $a_t$ when $G_{n,0}$ is conditioned on $a_0$.  The main tool of investigation is a generalized free energy function, which may be used to construct the desired predictions and analyzed to establish their validity.  The following section, Sec.\ \ref{sec:macroscopic_dynamics}, discusses some general properties of the macroscopic dynamics that may be inherited from the underlying microscopic dynamics, including time symmetry, equilibration, and the formation of a semigroup.  The final two sections examine particular applications.  In Sec.\ \ref{sec:extensive_noninteracting}, results from \cite{LaCour2000} for noninteracting systems are rederived and extended using our new approach, while Sec.\ \ref{sec:binary_lattice} contains a detailed study of a simple interacting lattice with an exact macroscopic solution.  Our results are summarized and discussed in Sec.\ \ref{sec:discussion}.

%%%%%%%%%%%%%%%%%%%%%%%%%%%%%%%%%%%%%%%%%%%%%%%%%%%%%%%%%%%%%%%%%%%%%%%%%%%%

\section{Problem Description}
\label{sec:problem_description}

Consider a microscopically deterministic system whose collective microstate is given by a point in a phase space $X_n$ and for which there is a Borel measurable function $\Phi_{n,t}$ which maps any microstate $x_0 \in X_n$ at time zero to its future/past value, $\Phi_{n,t}(x_0)$, at time $t \in \mathbb{R}$.  We shall assume for simplicity that the microscopic dynamics form a group; i.e., that $\Phi_{n,t}\circ\Phi_{n,s} = \Phi_{n,t+s}$.  On this phase space we further suppose there is given a Borel probability measure $P_n$, invariant under $\Phi_{n,t}$, which may be interpreted as the \textit{a priori} distribution of an ensemble of initial microstates.  The choice of an \textit{a priori} measure is a familiar yet subtle problem which shall not be addressed.  Here we merely assume that a suitable \textit{a priori} measure has been found.  It should be noted, however, that the restriction to dynamical groups and invariant measures, while physically reasonable, is not an essential one and may be relaxed with only minor changes in the results.

Now consider a measurable, vector-valued function $G_n: X_n \rightarrow M$, where $M \subseteq \mathbb{R}^d$ and $d$ is a positive integer.  The function $G_n$ is chosen to represent an arbitrary set of macroscopic variables of interest.  Given an initial \emph{microstate} $x_0$ (``initial'' here refers to time zero), the initial \emph{macrostate} is $G_n(x_0)$ and, similarly, the macrostate at time $t$ is $G_{n}(\Phi_{n,t}(x_0))$.  We therefore define $G_{n,t} \equiv G_n\circ\Phi_{n,t}$ as the macroscopic observable at time $t$.  (Note that $G_{n,0} = G_n$ from the group property.)

Clearly, $G_{n,t}$ is a deterministic function of the initial microstate.  The problem we wish to consider, however, is the following: given an initial \emph{macrostate} $a_0$ such that $G_{n,0}(x_0) = a_0$ for some \emph{unspecified} microstate $x_0$, what is the expected macrostate, $a_t$, at time $t$?  The term ``expected'' is here meant in the sense of an overwhelmingly most probable value, so we have now introduced the notion of a distribution of initial microstates corresponding to the given constraint on the initial macrostate.  This distribution derives from the assumed \textit{a priori} distribution $P_n$ and, for the above constraint, takes the form of a conditional probability $P_{n,a_0}$, assuming that a regular conditional distribution does indeed exist.  (Later we shall circumvent questions of regularity by conditioning on a set of nonzero probability.)

Considering $n$ as a macroscopic index representing, say, the number of particles in the system, we may consider the asymptotic behavior of $P_{n,a_0}$ as $n\rightarrow\infty$.  In particular, we expect that, if a law of large numbers holds, $a_t$ should be given by the limiting expectation value of $G_{n,t}$ with respect to $P_{n,a_0}$, i.e.,
\begin{equation}
a_t = \lim_{n\rightarrow\infty} \int_{X_n} G_{n,t} \; \d P_{n,a_0}.
\end{equation}
Provided a suitable sequence of conditional probability measures can be constructed for which the law of large numbers holds, the above asymptotic conditional expectation gives the desired mapping $a_0 \mapsto a_t$ from the initial macrostate to its value at time $t$.

The difficulty with this approach is, of course, in constructing the above conditional probability measures.  We therefore use a different approach in which one considers the joint distribution of $(G_{n,0},G_{n,t})$ under the \textit{a priori} measure $P_n$ and conditions on a set $B_0$ of nonzero probability such that $\left.(G_{n,0},G_{n,t})\right|_{B_0}$ converges to $(a_0,a_t)$ in probability, where $a_0$ is given and $a_t$ is to be determined.  To illustrate this approach, suppose $(G_{n,0},G_{n,t})$ converges in probability to the equilibrium value $(a_*,a_*)$, where
\begin{equation}
a_* \equiv \lim_{n\rightarrow\infty} \int_{X_n} G_n \; \d P_n.
\end{equation}
(Recall that $P_n$ is invariant.)  This situation is illustrated in Fig.\ \ref{fig:contours} for scalar-valued observables, where the contour plot shows a fiducial joint density of $G_{n,0}$ and $G_{n,t}$.  If we wish to condition on a value $a_0$ for $G_{n,0}$, a half-plane $B_0$ may be constructed whose boundary is the vertical line passing through $a_0$.  If the joint variables $(G_{n,0},G_{n,t})$ are restricted to $B_0$, and $B_0$ is of nonzero measure, then it is clear that $G_{n,0}$ will converge to $a_0$ in probability as $n\rightarrow\infty$, while $G_{n,t}$ converges to a value $a_t$ which, if the density level sets are convex, will uniquely maximize the probability along the vertical line through $a_0$.

%%%%%%%%%%%%%%%%%%%%%%%%%%%%%%%%%%%%%%%%%%%%%%%%%%%%%%%%%%%%%%%%%%%%%%%%%%%%

\section{Conditional LDP Approach Using the Dynamical Free Energy}
\label{sec:ldp_approach}

As we are interested only in asymptotically most likely behavior of the observables, an alternative to constructing density level curves is to use a level-1 large deviation rate function \cite{Ellis}.  The rate function governs convergence in probability in a manner similar to that of the entropy in the thermodynamic limit as described by the Boltzmann-Einstein fluctuation formula and discussed, e.g., by Lanford \cite{Lanford1973} and Martin-L\"{o}f \cite{Martin-Lof}.

% Joint, Unconditioned LDP

The rate function is most easily obtained from the following generalized free energy function, here referred to as the \emph{dynamic free energy}.
\begin{equation}
\Psi_t(\lambda_1,\lambda_2) \equiv \lim_{n\rightarrow\infty} v_n^{-1} \log \int_{X_n} \e^{v_n[ \lambda_1 G_{n,0} + \lambda_2 G_{n,t}]} \; \d P_n,
\end{equation}
where $\lambda_1$ and $\lambda_2$ are row vectors in $\mathbb{R}^d$ and the $v_n$ are positive scale factors such that $v_n \rightarrow \infty$ as $n\rightarrow\infty$.  (The values of $G_n$ are taken to be column vectors.)  Provided $\Psi_t$ is well defined, i.e., that the limit exists everywhere, it will be a convex function.  In what follows, we shall assume $\Psi_t$ is everywhere well defined, finite, and differentiable.  This, in fact, is the main physical assumption we make regarding the observable, $G_{n}$, the microscopic dynamics, as given by $\Phi_{n,t}$, and the choice of scaling, $v_n$.  An important theorem due to G\"{a}rtner and Ellis \cite{Dembo_and_Zeitouni} now implies that $(G_{n,0},G_{n,t})$ satisfies a large deviation principle (LDP) with a rate function $I_t$  given by the convex conjugate (in this case, the Legendre transform) of $\Psi_t$.

The rate function thus obtained is essentially strictly convex, meaning that it is strictly convex on the interior of the domain over which it is finite.  (We assume that $M$ is properly chosen so that $I_t$ is everywhere finite on $M \times M$.)  Thus, the rate function possesses a unique global minimum of zero and is positive everywhere else.  The point at which this global minimum is attained represents the equilibrium or \textit{a priori} most-probable macrostates for $G_{n,0}$ and $G_{n,t}$.  From the Legendre transform relation between $\Psi_t$ and $I_t$, we have
\begin{equation}
I_t(a_1,a_2) = \lambda_1 a_1 + \lambda_2 a_2 - \Psi_t(\lambda_1,\lambda_2),
\end{equation}
where $\lambda_1$ and $\lambda_2$, the \emph{conjugate macrostates}, satisfy the following:
\begin{subequations}\label{eqn:conjugate-to-macrostate}
\begin{align}
a_1 &= \nabla_1 \Psi_t(\lambda_1,\lambda_2), \\
a_2 &= \nabla_2 \Psi_t(\lambda_1,\lambda_2).
\end{align}
\end{subequations}
Taking $\lambda_1 = \lambda_2 = 0$ gives a value of zero for the rate function, so the equilibrium point is $(a_*,a_*)$, where
\begin{equation}
a_* = \nabla_1\Psi_t(0,0) = \lim_{n\rightarrow\infty} \int_{X_n} G_{n} \; \d P_n = \nabla_2\Psi_t(0,0),
\end{equation}
since $P_n$ is invariant.  Note that convexity of $\Psi_t$ permits the interchange of the limit and gradient operations.  (See Theorem 25.7 of \cite{Rockafellar}.)

The large deviation principle now implies that $(G_{n,0},G_{n,t})$ converges in probability to the equilibrium point $(a_*,a_*)$ as $n\rightarrow\infty$.  A point $(a_1,a_2)$ different from this equilibrium point corresponds to a fluctuation and, loosely speaking, has a probability which vanishes as $\exp[-v_n I_t(a_1,a_2)]$.  (We stress, however, that such a fluctuation refers to the variability in the unknown initial microstate and is not a temporal fluctuation.)  

It must again be emphasized that the above results hold under the critical assumption that the dynamic free energy is everywhere \emph{well defined}, \emph{finite}, and \emph{differentiable}.  The validity of these assumptions depends on the microscopic dynamics, as given by $\Phi_{n,t}$ in $G_{n,t} = G_{n}\circ\Phi_{n,t}$, as well as the choice of observables, $G_{n}$, scaling, $v_n$, and \textit{a priori} probability, $P_n$.  The one-to-one correspondence between macrostates $(a_1,a_2)$ and their conjugates $(\lambda_1,\lambda_2)$ is another important assumption which will be employed in the following discussion to establish a conditional LDP.  Note that this property will hold provided the Jacobian of $(\nabla_1\Psi_t, \nabla_2\Psi_t)$ exists everywhere and vanishes nowhere.  In Section \ref{sec:extensive_noninteracting}, an explicit example will be considered in which all these assumptions may be verified explicitly and rigorously.

% Conditional LDP

Given the sequence $\seq{(G_{n,0},G_{n,t})}$ of random vectors which satisfy the LDP with rate function $I_t$, we may construct a related sequence of random vectors, $\seq{\left.(G_{n,0},G_{n,t})\right|_{B_0}}$, by conditioning on a subset $B_0$ of $M \times M$.  Provided $I_t$ and $B_0$ are suitable, this sequence will also satisfy the LDP.  Specifically, the following general theorem is proven in Appendix \ref{app:Conditional_LDP}.

\begin{theorem}[Conditional LDP]
\label{thm:Conditional_LDP}
Suppose the sequence $\seq{P_n}$ of probability measures satisfies the LDP with a good, finite, and continuous rate function $I$ on a Hausdorff space.  Let $B$ denote a nonempty Borel subset such that $\closure{\interior{B}} = \closure{B}$.  Then the sequence $\seq{P_n(\,\cdot\,|B)}$ of conditional probability measures satisfies the LDP with the following good, finite, and continuous rate function.
\begin{equation}
I_B(x) = 
\begin{cases}
I(x) - \inf I(B) & \text{if $x \in \closure{B}$} \\
\infty           & \text{otherwise}
\end{cases}
\end{equation}
\end{theorem}

In the above notation, $\interior{B}$ is the interior of $B$, while $\closure{B}$ is its closure.  Clearly, $M \times M$ is a Hausdorff space under the relative topology of $\mathbb{R}^d$.  The G\"{a}rtner-Ellis theorem implies that $I_t$ is good and, as it is convex, continuous relative to $M \times M$ (Rockafellar \cite{Rockafellar}, Theorem 10.1).  For the condition $\closure{\interior{B_0}} = \closure{B_0}$ to hold, it suffices for $B_0$ to be convex with a nonempty interior (Rockafellar \cite{Rockafellar}, Theorem 6.3).  Satisfaction of the LDP for $\seq{(G_{n,0},G_{n,t})}$ comes, of course, from the assumed properties of the free energy $\Psi_t$.

To condition on a given value, $a_0$, of the initial observable $G_{n,0}$, we consider a conditional LDP with the conditioning set $B_0 \subset M \times M$ chosen so that the image, $I_t(B_0)$, of $B_0$ under $I_t$ has its infimum at a unique point $(a_0,a_t)$ for some $a_t$.  To this end we choose
\begin{equation}
\label{eqn:conditioning_set}
B_0 = \set{a\in M: \lambda_0 (a-a_0) \ge 0} \times M,
\end{equation}
where $\lambda_0$ satisfies $a_0 = \nabla\Psi(\lambda_0)$ and $\Psi(\cdot) \equiv \Psi_t(\cdot,0)$.  (We assume that $\lambda_0$, so defined, exists and is unique.)  From its definition we see that $B_0$ describes a nonempty, convex half space demarcated by the hyperplane containing $(a_0,a_t)$ with normal vector $(\lambda_0,0)$ directed toward its interior.  We shall assume that $a_0$ is well chosen so that $\interior{B_0}$ is nonempty and hence the conditions of Theorem \ref{thm:Conditional_LDP} are satisfied.  Since $G_{n,0}$ is expected to converge in probability to $a_0$, we shall denote by $G_{n,0} \approx a_0$ the event $(G_{n,0},G_{n,t}) \in B_0$.

In Appendix \ref{app:Conditional_Infimum} it is shown that
\begin{equation}
\inf I_{t}(B_0) = I_t(a_0,a_t) = \lambda_0 a_0 - \Psi_t(\lambda_0,0),
\end{equation}
where $a_t = \nabla_2\Psi_t(\lambda_0,0)$ and $(a_0,a_t)$ is the unique such point.  We therefore conclude that $\left.(G_{n,t_0},G_{n,t})\right|_{B_0}$ satisfies the LDP with the good, effectively continuous, essentially strictly convex rate function $I_{t,B_0}$, where
\begin{equation}
I_{t,B_0}(a_1,a_2) = 
\begin{cases}
I_{t}(a_1,a_2) - \lambda_0 a_0 + \Psi_t(\lambda_0,0) & \text{if $(a_1,a_2) \in \closure{B_0}$}, \\
\infty & \text{otherwise}.
\end{cases}
\end{equation}

The unique equilibrium point of $I_{t,B_0}$ is $(a_0,a_t)$, where $a_0$ specifies $B_0$ and $a_t$ is given by $\nabla_2\Psi_{t}(\lambda_0,0)$.  As the large deviation principle implies convergence in probability, we also have
\begin{equation}
\lim_{n\rightarrow\infty} P_n\left[(G_{n,0},G_{n,t}) \in B \,\right|\left.G_{n,0} \approx a_0\right] = 
\begin{cases}
1 & \text{if $(a_0,a_t) \in \interior{B}$,} \\
0 & \text{if $(a_0,a_t) \not\in \closure{B}$.}
\end{cases}
\end{equation}
Since $\Psi_{t}$ is a convex function which is assumed to be finite and differentiable everywhere, it follows that the gradients and limits as $n\rightarrow\infty$ may be interchanged (Rockafellar \cite{Rockafellar}, Theorem 25.7), allowing us to write
\begin{subequations}
\begin{align}
a_0 &= \nabla_1\Psi_{t}(\lambda_0,0) = \lim_{n\rightarrow\infty} Z_n(\lambda_0)^{-1} \int_{X_n} G_{n} \, \e^{v_n\lambda_0 G_{n}} \; \d P_n, \label{eqn:initial_macrostate} \\
a_t &= \nabla_2\Psi_{t}(\lambda_0,0) = \lim_{n\rightarrow\infty} Z_n(\lambda_0)^{-1} \int_{X_n} G_{n,t} \, \e^{v_n\lambda_0 G_{n}} \; \d P_n, \label{eqn:final_macrostate}
\end{align}
\end{subequations}
where
\begin{equation}
Z_n(\lambda_0) \equiv \int_{X_n} \e^{v_n\lambda_0 G_{n}} \; \d P_n.
\end{equation}
We therefore see that $a_0$ and $a_t$ are given by canonical expectations, where the canonical probability measure is given by
\begin{equation}
\d P_{n,\lambda_0} \equiv Z_n(\lambda_0)^{-1} \e^{v_n \lambda_0 G_n} \; \d P_n
\end{equation}
and $\lambda_0$ is determined from $a_0$.

% LDP for Conditional Observable

Using the contraction principle \cite{Dembo_and_Zeitouni} we may obtain the LDP for $\left.G_{n,t}\right|_{a_0}$, the observable at time $t$ conditioned on $G_{n,0} \approx a_0$.  To do this, note that the mapping $(a_1,a_2) \mapsto a_2$ is continuous, and therefore $\left.G_{n,t}\right|_{a_0}$ satisfies the LDP with rate function $I_{t,a_0}$, where
\begin{equation}
\begin{split}
I_{t,a_0}(a) &= \inf \set{I_{t,B_0}(a_1,a_2): a_1 \in M, a_2 = a} \\
&= \inf \set{I_{t}(a_1,a): a_1 \in M, \lambda_0 (a_1-a_0) \ge 0}  - \lambda_0 a_0 + \Psi(\lambda_0) \\
&= I_{t}(a_0,a) - \lambda_0 a_0 + \Psi(\lambda_0),
\end{split}
\end{equation}
for $a \in M$.  From the properties of $I_{t}$, $I_{t,a_0}$ is a good convex rate function.  Since $I_{t,a_0}$ is convex, the corresponding free energy, call it $\Psi_{t,a_0}$, is given by the convex conjugate of $I_{t,a_0}$ \cite{Dinwoodie1993}.  Thus,
\begin{equation}
\begin{split}
\Psi_{t,a_0}(\lambda) &= \sup_{a' \in M} [\lambda a' - I_{t,a_0}(a')] \\
&= \sup_{a' \in M} [\lambda_0 a_0 + \lambda a' - I_{t}(a_0,a') - \Psi(\lambda_0)] \\
&= \Psi_{t}(\lambda_0,\lambda) - \Psi(\lambda_0),
\end{split}
\end{equation}
for $\lambda \in \mathbb{R}^d$.  The properties of finiteness and differentiability for $\Psi_{t,a_0}$ are inherited from those of $\Psi_{t}$, or, more specifically, from those of $\Psi_{t}(\lambda_0,\cdot\,)$.  Furthermore if $\Psi_{t,a_0}$, as a function of complex variables, is analytic at the origin, then $\left.G_{n,t}\right|_{a_0}$ is asymptotically Gaussian with mean $\nabla\Psi_{t,a_0}(0)$ and covariance matrix $\nabla\transpose{\nabla}\Psi_{t,a_0}(0)$ \cite{Bryc1993}.  Irrespective of the analyticity of $\Psi_{t,a_0}$, the global minimum of $I_{t,a_0}$ is attained at $\nabla\Psi_{t,a_0}(0) = \nabla_2\Psi_{t}(\lambda_0,0) = a_t$, as expected.  This may be seen by noting that $I_{t,a_0}$ may be expressed in terms of $\Psi_{t,a_0}$ via the Legendre transform relation
\begin{equation}
I_{t,a_0}(a) = \lambda a - \Psi_{t,a_0}(\lambda), \quad \text{where $a = \nabla\Psi_{t,a_0}(\lambda)$}.
\end{equation}

It should be noted that the above results are identical to that obtained by more traditional information theoretic methods for nonequilibrium phenomena \cite{Jaynes1980}.  This should not be surprising, since both methods rely in essence upon an entropy maximization procedure.  The crucial difference and importance of the large deviation approach discussed here is that the maximum entropy principle is now given a rigorous, statistical justification; it need not be taken as a mere \textit{ad hoc} assumption, however reasonable.  The key difference is that, in information theory, constraints are placed based upon the ``known'' expectation values, whereas in the theory of large deviations the constraints are imposed upon the actual macroscopic observables.  The equivalence of these constraints, an implicit assumption of information theory, is explicitly proven through large deviation theory.

%%%%%%%%%%%%%%%%%%%%%%%%%%%%%%%%%%%%%%%%%%%%%%%%%%%%%%%%%%%%%%%%%%%%%%%%%%%%%%%

\section{Macroscopic Dynamics}
\label{sec:macroscopic_dynamics}

Using Eqs.\ (\ref{eqn:initial_macrostate}) and (\ref{eqn:final_macrostate}) one may define a map $\psi_t$ from a given initial macrostate $a_0$ at time zero to a final, predicted macrostate $\psi_t(a_0) = a_t$ at time $t\in\mathbb{R}$.  The formal definition of this map requires only that the corresponding integral expressions are well defined and that the relation between $a_0$ and its conjugate, $\lambda_0$, be invertible.  The LDP, however, serves to connect this predicted macrostate to a typical realization of the time-dependent observable $G_{n,t}$ by stating that the latter converges in probability to the former when conditioned on $G_{n,0} \approx a_0$.  It is in this sense that the map $\psi_t$ or, more precisely, the family of maps $\set{\psi_t}_{t\in\mathbb{R}}$ represents the macroscopic dynamics of the system.

For a given $a_0$, the curve $\set{\psi_t(a_0)}_{t\in\mathbb{R}}$ represents a set of most-probable macrostates which we have referred to in \cite{LaCour2000} as the deterministic curve.  Thus, for each given time $t$ the most probable macrostate at this time is $\psi_t(a_0)$.  This deterministic curve should be distinguished, however, from the actual trajectory, $\set{G_{n,t}(x_0)}_{t\in\mathbb{R}}$, of a particular realization for finite $n$ with initial microstate $x_0$.  As we have pointed out previously \cite{LaCour2000}, despite their similarities there may be striking qualitative differences between $\set{\psi_t(a_0)}_{t\in\mathbb{R}}$ and $\set{G_{n,t}(x_0)}_{t\in\mathbb{R}}$ owing, e.g., to the presence of Poincar\'{e} recurrences.

In this section we shall consider some properties of the macroscopic map, $\psi_t$, as a map on the set of macrostates and the deterministic curve, $\set{\psi_t(a_0)}_{t\in\mathbb{R}}$, as a function of time.

%------------------------------------------------------------------------------

\subsection{Group/Semigroup Property}

The group property of the microscopic dynamics, $\set{\Phi_{n,t}}_{t\in\mathbb{R}}$, does not necessarily imply a group nor even a semigroup property for the macroscopic dynamics, $\set{\psi_t}_{t\in\mathbb{R}}$.  This was previously found to be true in noninteracting systems \cite{LaCour2000}; we discuss its validity in the broader class of interacting systems.  To begin, it is important to distinguish conceptually between a quantity such as $\psi_{t+s}(a_0)$ on the one hand and $\psi_s(\psi_t(a_0))$ on the other.  The former is the expected macrostate at time $t+s$ given that the macrostate at time zero is $a_0$.  By contrast, the latter is the expected macrostate at time $t+s$ given only that the macrostate at time $t$ is $\psi_t(a_0)$, which in turn is the predicted macrostate at time $t$ given only that the macrostate at time zero is $a_0$.  While it is true that the macroscopic dynamics are time translation invariant, i.e., conditioning $G_{n,t+s}$ on $G_{n,t}$ is the same as conditioning $G_{n,s}$ on $G_{n,0}$, it need not be the case that $\psi_s(\psi_t(a_0)) = \psi_{t+s}(a_0)$.  The latter equality is, of course, equivalent to a group or semigroup property, neither of which need hold.

There are several properties of the deterministic curve which, if they hold, preclude the formation of a group.  For example, suppose $\psi_t(a_0) = a_*$ for some $a_0 \neq a_*$.  If the group property were to hold, then invariance of the \textit{a priori} measure would imply that $a_0 = \psi_{-t}(\psi_t(a_0)) = \psi_{-t}(a_*) = a_*$, which is a contradiction.  Time symmetry offers a similar obstacle since, supposing again the group property, the fact that $\psi_t = \psi_{-t} = \psi_t^{-1}$ for all $t\in\mathbb{R}$ would imply that $\psi_t = \psi_{t/2}\circ\psi_{t/2} = \psi_{t/2}\circ\psi_{t/2}^{-1}$ is the identity.  Furthermore, if the semigroup property does holds it will imply a monotonic convergence to equilibrium if $\psi_t(a_0)$ is never further from equilibrium that it is initially.  This may be seen by noting that $\abs{\psi_{t+s}(a_0)-a_*} \le \abs{a_0-a_*}$ and $\abs{\psi_s(\psi_t(a_0))-a_*} \le \abs{\psi_t(a_0)-a_*}$ would imply $\abs{\psi_{t+s}(a_0)-a_*} \le \abs{\psi_t(a_0)-a_*}$.

A useful necessary and sufficient condition for the semigroup property to hold may be obtained as follows:  First, we make the trivial observation that $\psi_{t+s}(a_0) = \psi_t(\psi_s(a_0))$ iff $\psi_{t+s}(a_0)\psi_s(a_0) = \psi_t(\psi_s(a_0))\psi_s(a_0)$.  Now, $\psi_s(a_0)$ may be interpreted in one of two ways.  As the predicted macrostate at time $s$ based on conditioning to the value $a_0$ at time zero, it obeys
\begin{equation}
\psi_s(a_0) = \lim_{n\rightarrow\infty} \int_{X_n} G_{n,s} \; \d P_{n,\lambda_0},
\end{equation}
where $\lambda_0$ corresponds to $a_0$ via Eq.\ (\ref{eqn:initial_macrostate}).  Considered as the initial condition in $\psi_t(\psi_s(a_0))$, however, $\psi_s(a_0)$ satisfies
\begin{equation}
\psi_s(a_0) = \lim_{n\rightarrow\infty} \int_{X_n} G_{n} \; \d P_{n,\lambda_s},
\end{equation}
which relates $\psi_s(a_0)$ to its conjugate macrostate, $\lambda_s$.  The semigroup property will therefore hold iff
\begin{multline}
\frac{\psi_{t+s}(a_0)\psi_s(a_0)}{\psi_t(\psi_s(a_0))\psi_s(a_0)} \\
= \lim_{n\rightarrow\infty} \frac{\int_{X_n}G_{n,t+s}\,\e^{n\lambda_0G_n}\;\d P_n \;\cdot\; \int_{X_n}G_{n}\,\e^{n\lambda_s G_n}\;\d P_n}{\int_{X_n}G_{n,t}\,\e^{n\lambda_sG_n}\;\d P_n \;\cdot\; \int_{X_n}G_{n,s}\,\e^{n\lambda_0G_n}\;\d P_n} = 1
\label{eqn:semigroup_condition}
\end{multline}
for all $\lambda_0$ and $0 \le s \le t$ (or $t \le s \le 0$).

%------------------------------------------------------------------------------

\subsection{Time Symmetry}

From the group property, the dynamical map $\Phi_{n,t}$ is not only invertible but also time reversible; i.e., $\Phi_{n,t}^{-1} = \Phi_{n,-t}$.  Suppose furthermore that the dynamics are time reversal invariant; i.e., $\Phi_{n,-t} = R_n\circ\Phi_{n,t}\circ R_n$ for some involution $R_n$.  (For Hamiltonian systems, e.g, $R_n$ corresponds to momentum reversal.)  When this is the case, every microstate $x \in X_n$ has a mirror state $R_n(x)$ such that $\Phi_{n,t}(x) = R_n(\Phi_{n,-t}(R_n(x)))$; hence, the trajectory of $x$ is mirrored by the trajectory of $R_n(x)$ with the direction of time reversed.  Suppose the set $\set{x\in X_n: R_n(x)=x}$ forms a boundary separating two disjoint regions of microstates and their mirror states.  If the observable, $G_n$, and the \textit{a priori} measure, $P_n$, are both invariant under $R_n$, then a typical sample of initial microstates such that $G_{n,0} \approx a_0$ will include roughly equal numbers of points from these two regions.  On this basis, one expects the trajectories $\set{G_{n,t}}_{t\le0}$ and $\set{G_{n,t}}_{t\ge0}$ to be \emph{approximately} symmetric when $n$ is large.  This suggests that the expected trajectories $\set{\psi_t(a_0)}_{t\le0}$ and $\set{\psi_t(a_0)}_{t\ge0}$ should be \emph{perfectly} symmetric in time.  That this is indeed the case may be seen by noting the following:
\begin{equation*}
\begin{split}
\psi_{t}(a_0) &= \lim_{n\rightarrow\infty} Z_n(\lambda_0)^{-1} \int_{X_n}(G_n\circ\Phi_{n,t}) \, \e^{n\lambda_0G_n}\,\d P_n \\
&= \lim_{n\rightarrow\infty} Z_n(\lambda_0)^{-1} \int_{X_n}(G_n\circ R_n\circ\Phi_{n,t}) \, \e^{n\lambda_0G_n\circ R_n}\,\d P_n \\
&= \lim_{n\rightarrow\infty} Z_n(\lambda_0)^{-1} \int_{X_n}(G_n\circ\Phi_{n,-t}\circ R_n) \, \e^{n\lambda_0G_n\circ R_n}\,\d P_n \\
&= \lim_{n\rightarrow\infty} Z_n(\lambda_0)^{-1} \int_{X_n}(G_n\circ\Phi_{n,-t}) \, \e^{n\lambda_0G_n}\,\d (P_n\circ R_n^{-1}) \\
&= \psi_{-t}(a_0).
\end{split}
\end{equation*}

A common example is that of a Hamiltonian system with an observable which depends only on the positions of the microscopic constituents and not on their momenta.  It is important to note that the symmetry is with respect to the point in time at which the initial condition is imposed.  As noted above, this symmetry precludes a group property.  If a semigroup property holds, however, then there are in fact two semigroups, which are identical in all but the orientation of the time axis.

%------------------------------------------------------------------------------

\subsection{Continuity and Equilibration}

If $G_n$ is a discontinuous function, then specific realizations of $\set{G_{n,t}}_{t\in\mathbb{R}}$ will also be discontinuous in $t$.  However, as $\psi_t(a_0)$ is given by an limiting integral expression, we would expect that $\psi_t(a_0)$ is continuous in $t$, provided $G_{n,t}$ is almost-everywhere continuous in $t$, by Lebesgue dominated convergence.  The difficulty, of course, is in the interchange of limits $t'\rightarrow t$ and $n\rightarrow\infty$.  If, however, $\int_{X_n} G_{n,t} \; \d P_{n,\lambda_0}$ converges \emph{uniformly} in $t$ as $n\rightarrow\infty$, then we will have the desired continuity result.

A similar problem is encountered in considering the long-time behavior of $\psi_t(a_0)$.  If $\Phi_{n,t}$ is, say, mixing with an asymptotic measure equal to the invariant \textit{a priori} measure $P_n$, then
\begin{equation}
\lim_{n\rightarrow\infty}\lim_{t\rightarrow\infty} \int_{X^n}G_{n,t}\,\d P_{n,\lambda_0} = \lim_{n\rightarrow\infty}\int_{X^n}G_n\,\d P_n \equiv a_*,
\end{equation}
but
\begin{equation}
\lim_{t\rightarrow\infty} \lim_{n\rightarrow\infty} \int_{X^n}G_{n,t}\,\d P_{n,\lambda_0} = \lim_{t\rightarrow\infty} \psi_t(a_0).
\end{equation}
Again, uniform convergence permits the interchange of these limits.  In this case, we say that the initial macrostate, $a_0$, \emph{equilibrates} to the \textit{a priori} most likely or equilibrium value, $a_*$, since $\psi_t(a_0) \rightarrow a_*$ as $t \rightarrow \infty$.  The mixing property implies that this convergence, if it holds, is independent of $a_0$.

Although long time convergence has been occasionally described as ``time asymmetric'' behavior (because of the difference in macrostates between $t=0$ and $t=\infty$), in truth the convergence may occur in both the forward and reverse time directions.  This will be true, for example, if the macroscopic map is time symmetric.  Furthermore, this convergence need not be monotonic.  An exception, however, is when the family of macroscopic maps forms a semigroup.

%------------------------------------------------------------------------------

\subsection{Affine Covariance}

We end with a brief discussion of affine covariance.  Given the observables $G_{n,0}$ and $G_{n,t}$ which satisfy the regularity conditions for $\Psi_t$, consider the affine transformed observable
\begin{equation}
G_n' \equiv A G_n + b,
\end{equation}
where $A$ is a nonsingular matrix and $b$ is a column vector.  The dynamic free energy for the transformed pair is
\begin{equation}
\begin{split}
\Psi_{t}'(\lambda_1',\lambda_2') &\equiv \lim_{n\rightarrow\infty} v_n^{-1} \log \int_{X_n} \e^{v_n[\lambda_1' G_{n,0}' + \lambda_2' G_{n,t}']} \; \d P_n \\
&= \Psi_{t}(\lambda_1 A', \lambda_2' A) + {\lambda_1' b} + {\lambda_2' b}.
\end{split}
\end{equation}
Clearly $\Psi_{t}'$ inherits all of the differentiability properties of $\Psi_{t}$.  Conditioning on $G_{n,0}' = a_0'$, we therefore obtain the following set of equations:
\begin{subequations}
\begin{align}
a_0' &= \nabla_1'\Psi_{t}'(\lambda_0',0) = A \nabla_1\Psi_{t}(\lambda_0' A,0) + b, \\
a_t' &= \nabla_2'\Psi_{t}'(\lambda_0',0) = A \nabla_2\Psi_{t}(\lambda_0' A,0) + b.
\end{align}
\end{subequations}
Letting $\lambda_0 = \lambda_0'A$ in Eq.\ (\ref{eqn:final_macrostate}) gives $\psi_{t}(a_0) = \nabla_2\Psi_{t}(\lambda_0'A,0)$, where $a_0 = \nabla_1\Psi_{t}(\lambda_0'A,0)$ by Eq.\ (\ref{eqn:initial_macrostate}).  Since $A$ is invertible, $a_0 = A^{-1}(a_0'-b)$ and the transformed macroscopic map is
\begin{equation}
\psi_{t}'(a_0') = A \psi_{t}(A^{-1}(a_0'-b)) + b.
\label{eqn:affine_covariance}
\end{equation}
Thus, for any invertible affine transformation $T$, $\psi_{t}' = T\circ\psi_{t}\circ T^{-1}$.

Clearly continuity, time symmetry, and convergence are affine invariant properties of the macroscopic map.  The semigroup property is also affine invariant since, if $T$ is any invertible affine transformation,
\begin{equation}
\begin{split}
\psi_{t+s}' &= T\circ\psi_{t+s}\circ T^{-1} = T\circ(\psi_t\circ\psi_s)\circ T^{-1} \\
&= (T\circ\psi_t\circ T^{-1})\circ(T\circ\psi_s\circ T^{-1}) = \psi_t'\circ\psi_s'.
\end{split}
\end{equation}

%%%%%%%%%%%%%%%%%%%%%%%%%%%%%%%%%%%%%%%%%%%%%%%%%%%%%%%%%%%%%%%%%%%%%%%%%%%%%%%

\section{Extensive Noninteracting Systems}
\label{sec:extensive_noninteracting}

Previously in \cite{LaCour2000} we considered systems of $n$ dynamically independent and identical constituents with extensive observables of the form
\begin{equation}
G_n(x_1,\ldots,x_n) = \frac{1}{n}\sum_{i=1}^{n}g(x_i),
\end{equation}
where $(x_1,\ldots,x_n) \in X^n = X_n$ and $g$ is a bounded, measurable function from $X$ to $M$.  (The former restrictions that $g$ be almost-everywhere continuous and scalar valued are here lifted.)  The microscopic dynamics are given by
\begin{equation}
\Phi_{n,t}(x_1,\ldots,x_n) = \left(\,\varphi_t(x_1),\ldots,\varphi_t(x_n)\,\right),
\end{equation}
where $\varphi_t$ is a measurable transformation on $X$.  The \textit{a priori} probability measure is the product measure $P_n = P^n$, so all constituents are considered to be \textit{a priori} statistically independent and identical.  In this section, the results from \cite{LaCour2000} will be rederived using the present approach and, due to the lessening of restrictions on $g$, generalized.

Consider first the dynamic free energy.  Taking $v_n = n$, this is given by
\begin{equation}
\begin{split}
\Psi_t(\lambda_1,\lambda_2) &= \lim_{n\rightarrow\infty} \frac{1}{n}\log\int_{X^n}\e^{\lambda_1\sum_i g(x_i)+\lambda_2\sum_j g(\varphi_t(x_j))} \; \d P^n(x_1,\ldots,x_n) \\
&= \log \int_{X} \e^{\lambda_1 g(x) + \lambda_2 g(\varphi_t(x))} \; \d P(x).
\end{split}
\end{equation}
Since the limit clearly converges, $\Psi_{n,t}$ is indeed well defined.  Boundedness of $g$ implies $\Psi_{n,t}$ is everywhere finite, since
\begin{equation}
\begin{split}
|\e^{\Psi_t(\lambda_1,\lambda_2)}| &\le \int_{X} \abs{\e^{\lambda_1 g(x) + \lambda_2 g(\varphi_t(x))}} \; \d P(x) \\
&\le \int_{X} \e^{\lambda_1 \norm{g}_{\infty} + \lambda_2 \norm{g}_{\infty}} \; \d P(x) < \infty.
\end{split}
\end{equation}
Finally, Lebesgue dominated convergence implies $\Psi_t$ is everywhere differentiable.  Thus, $(G_{n,0},G_{n,t})$ satisfies the LDP with the good, convex rate function
\begin{equation}
I_t(a_1,a_2) = \lambda_1a_1 + \lambda_2a_2 - \Psi_t(\lambda_1,\lambda_2),
\end{equation}
where $\lambda_1$ and $\lambda_2$ are the conjugate macrostates corresponding to $a_1$ and $a_2$, respectively.  (In this case, the LDP follows more directly from an application of Cram\'{e}r's theorem \cite{Dembo_and_Zeitouni}.)

Conditioning on $G_{n,0} \approx a_0$ leads to a conditional LDP, as described in Sec.\ \ref{sec:ldp_approach}, with asymptotically most likely values for $G_{n,0}$ and $G_{n,t}$ of $a_0$ and $\psi_t(a_0)$.  From Eqs.\ (\ref{eqn:initial_macrostate}) and (\ref{eqn:final_macrostate}), we find by differentiating $\Psi_t$ that
\begin{subequations}
\begin{align}
a_0 &= \int_X g(x)\,\e^{\lambda_0 g(x)} \; \d P(x) / \int_{X} \e^{\lambda_0 g(x)} \; \d P(x) \\
\psi_t(a_0) &= \int_X g(\varphi_t(x))\,\e^{\lambda_0 g(x)} \; \d P(x) / \int_{X} \e^{\lambda_0 g(x)} \; \d P(x).
\end{align}
\end{subequations}
This recovers the results from \cite{LaCour2000} with the added generalization that $g$ need not be almost-everywhere continuous nor scalar in value.

Invertibility of the transformation $\lambda_0 \mapsto a_0 = \nabla\Psi(\lambda_0)$ may be determined by considering the Jacobian matrix $J = \nabla\transpose{\nabla}\Psi$, where
\begin{equation}
J(\lambda_0) = \frac{\int g\transpose{g}\e^{\lambda_0 g}\;\d P}{\int \e^{\lambda_0 g}\;\d P} - \frac{\int g\,\e^{\lambda_0 g}\d P \int \transpose{g}\e^{\lambda_0 g}\d P}{\left(\int \e^{\lambda_0 g}\;\d P\right)^2}.
\end{equation}
Thus, $J(\lambda_0)$ is just the covariance matrix of $g$ with respect to the canonical distribution corresponding to $\lambda_0$.  If $g = c$ is almost everywhere a constant, then $J$ will be identically zero and the transformation is clearly not invertible.  On the other hand, suppose $g = 1_{C}$, where $1_{C}$ is the indicator function; i.e., $1_{C}(x) = 1$ if $x\in C$ and $1_{C}(x) = 0$ otherwise.  In this case,
\begin{equation}
J(\lambda_0) = \frac{\e^{\lambda_0}P(C)[1-P(C)]}{\left[1-P(C)+\e^{\lambda_0}P(C)\right]^2},
\end{equation}
which is zero if and only if $P(C)$ is zero or one, which in either case implies $g$ is constant almost everywhere.  A nonzero Jacobian then entails that the transformation is invertible.

%%%%%%%%%%%%%%%%%%%%%%%%%%%%%%%%%%%%%%%%%%%%%%%%%%%%%%%%%%%%%%%%%%%%%%%%%%%%%%%

\section{Interacting Binary-State Lattice}
\label{sec:binary_lattice}

We now turn to an example of a simple interacting system.  Consider a one-dimensional lattice $s = (s_0,\ldots,s_{n-1})$ of $n$ sites, each site having a ``spin'' value 1 or 0 (``up'' or ``down'').  Such a system was the subject of an early study by Wolfram \cite{Wolfram1983} on cellular automata and provides a crude model of interacting magnetic spins.  (Sutner \cite{Sutner1988} has considered generalizations of this problem using graph-theoretic methods.)  We shall assume all sites are dynamically identical, which implies periodic boundary conditions.  For nearest-neighbor interactions, the discrete-time dynamics of the lattice are governed by a map $\phi$, where $s \mapsto \Phi_{n,t}(s) = \Phi_{n,1}^t(s)$ and
\begin{equation}
\pi_j(\Phi_{n,1}(s)) = \phi(s_{[j-1]_n},\,s_j,\,s_{[j+1]_n})
\end{equation}
for $j=0,\ldots,n-1$.  Here, $\pi_j$ is the canonical projection and $[\,\cdot\,]_n$ indicates the modulo $n$ function, whose use imposes periodic boundary conditions.  The \textit{a priori} distribution of the spins is taken to be independent and uniform.  Thus, each possible initial lattice has probability $1/2^n$.

For the dynamics we shall use Wolfram's ``rule 90,'' in which the state of a site at the next time step is the sum (modulo 2) of its two neighbors.  Symbolically, this means
\begin{equation}
\phi(s_0,s_1,s_2) = [s_0+s_2]_2^{}.
\end{equation}
In this map, one neighbor up and the other down leads to an up state, while two neighboring down states lead to a down state.  This rule does not correspond well to a ferromagnet, though, since two neighboring up states also lead to a down state.  For this reason, the map $\Phi_{n,t}$ is neither invertible nor measure preserving.  (Note, e.g., that both $(0,\ldots,0)$ and $(1,\ldots,1)$ map to $(0,\ldots,0)$.)  This model should therefore not be viewed as a dynamic Ising model but merely as a simple, yet interesting, cellular automaton in its own right.

Taking the observable to be the mean spin or ``magnetization,'' the dynamic free energy is
\begin{equation}
\Psi_t(\lambda_1,\lambda_2) = \lim_{n\rightarrow\infty} \frac{1}{n} \log \left[ \frac{1}{2^n}\sum_{s_0}\cdots\sum_{s_{n-1}} \e^{\lambda_1\sum_i s_i + \lambda_2\sum_j \pi_j(\Phi_{n,t}(s_0,\ldots,s_{n-1}))}\right].
\end{equation}
The finite-$n$ dynamic free energy, $\Psi_{n,t}$, is shown in Fig.\ \ref{fig:binary-free} for $t=1$ and $n=10$.  As $n$ is finite, this is only an approximation to the true dynamic free energy, $\Psi_t(\lambda_1,\lambda_2)$, but the convergence appears to be fairly rapid.  The surface is smooth and regular, supporting our assumption that $\Psi_t$ is everywhere finite and differentiable.  This allows us to identify $\psi_t(a_0) = \nabla_2\Psi_t(\lambda_0,0)$, where $a_0 = \nabla_1\Psi_t(\lambda_0,0)$, as the asymptotically most probable value of $G_{n,t}$ conditioned on $G_n \approx a_0$ for large $n$.  The graph of $\Psi_{n,t}$ for other values of $t$ is qualitatively quite similar.

Given an initial magnetization of $a_0$, we find that the corresponding conjugate macrostate, $\lambda_0$, is given by
\begin{equation}
\begin{split}
a_0 &= \lim_{n\rightarrow\infty} \sum_{s_0}\cdots\sum_{s_{n-1}} {\frac{1}{n} \sum_j s_j \, \e^{\lambda_0\sum_i s_i}} / \sum_{s_0}\cdots\sum_{s_{n-1}} \e^{\lambda_0\sum_i s_i} \\
&= \sum_{s_0}s_0\,\e^{\lambda_0 s_0} / \sum_{s_0}\e^{\lambda_0 s_0} = \frac{\e^{\lambda_0}}{1+\e^{\lambda_0}},
\end{split}
\end{equation}
or $\lambda_0 = \log[a_0/(1-a_0)]$.  Not surprisingly, this is the same result obtained in \cite{LaCour2000} for fractional occupations.

The macroscopic map itself is given formally by
\begin{equation}
\psi_t(a_0) = \lim_{n\rightarrow\infty} \frac{1}{n} \sum_{s_0}\cdots\sum_{s_{n-1}} \sum_j\pi_j\left(\Phi_{n,t}(s_0,\ldots,s_{n-1})\right) \prod_{i=0}^{n-1}\frac{\e^{\lambda_0s_i}}{1+\e^{\lambda_0}}.
\end{equation}
Since all sites are dynamically identical, $\sigma_j\circ\Phi_{n,t} = \Phi_{n,t}\circ\sigma_j$, where $\sigma_j$ is the shift map defined by $\pi_j(\sigma_j(s)) = s_{[i+j]_2}$.  In other words, the interactions are translation invariant and the lattice boundary conditions are periodic --- rotating the arguments is equivalent to rotating the site labels.  Due to nearest neighbor interactions, the state of any given lattice site $j$ at time $t$ will depend upon itself and its $2t$ neighboring sites.  Thus, $\pi_j\circ\Phi_{n,t}$ may be written
\begin{equation}
\pi_j\circ\Phi_{n,t} = \varphi_t\circ\pi_{0,\ldots,2t}\circ\sigma_{j-t},
\end{equation}
where $\pi_{0,\ldots,2t}$ is the canonical projection from $\set{0,1}^n$ to $\set{0,1}^{2t+1}$ (assuming $n \ge 2t+1$) and $\varphi_t(s_0,\ldots,s_{2t})$ is the state of site $t$ at time $t$ given that the initial states of sites $0$ through $2t$ are $s_0,\ldots,s_{2t}$, respectively.  We therefore have
\begin{equation}
\psi_t(a_0) = \lim_{n\rightarrow\infty} \frac{1}{n}\sum_j \sum_{s_0}\cdots\sum_{s_{n-1}} (\varphi_t\circ\pi_{0,\ldots,2t}\circ\sigma_{j-t}) \prod_{i=0}^{n-1}\frac{\e^{\lambda_0s_i}}{1+\e^{\lambda_0}}.
\end{equation}
By reordering the summation variables, the shift map $\sigma_{j-t}$ may be transferred to the product, which thereby remains unchanged.  Since only $2t+1$ sites are involved in the sum, we find
\begin{equation}
\begin{split}
\psi_t(a_0) &= \sum_{s_0}\cdots\sum_{s_{2t}} \varphi_t(s_0,\ldots,s_{2t})\,\prod_{i=0}^{2t} \frac{\e^{\lambda_0 s_i}}{1+\e^{\lambda_0}} \\
&= \sum_{s_0}\cdots\sum_{s_{2t}} \varphi_t(s_0,\ldots,s_{2t})\,\prod_{i=0}^{2t}a_0^{s_i}(1-a_0)^{1-s_i}.
\end{split}
\end{equation}
Finally, $\varphi_t$ may be defined recursively as follows: $\varphi_0(s_0) = s_0$ and
\begin{equation}
\varphi_t(s_0,\ldots,s_{2t}) = \left[\varphi_{t-1}(s_0,\ldots,s_{2t-2}) + \varphi_{t-1}(s_2,\ldots,s_{2t}) \right]_2.
\end{equation}

Wolfram \cite{Wolfram1983} notes that, if $s_{t}$ is the only nonzero site, the lattice at time $t\le (n-1)/2$ is given by the binomial coefficients, modulo 2.  Thus,
\begin{equation}
\Phi_{n,t}(s) = \left(\;s_0\left[\binom{t}{0}\right]_2\!\!,\, 0,\, s_2\left[\binom{t}{1}\right]_2\!\!,\, 0,\, \ldots,\, s_{2t}\left[\binom{t}{t}\right]_2\!\!,\, 0,\, \ldots,\, 0\;\right).
\end{equation}
An analogous rule holds if any other $s_i$ is the only nonzero site.  For an arbitrary initial lattice, Wolfram notes that the final lattice is given by the additive superposition, modulo 2, of these $n$ basic lattice states.  From this ``additive superposition'' property we find
\begin{equation}
\varphi_t(s_0,\ldots,s_{2t}) = \left[ \sum_{j=0}^{t}s_{2j}\left[\binom{t}{j}\right]_2 \right]_2
\end{equation}
for an arbitrary $s$.  Separating even and odd indices, the macroscopic map is found to be given by
\begin{equation}
\begin{split}
\psi_t(a_0) &= \sum_{s_0}\cdots\sum_{s_{2t}} \left[ \sum_{j=0}^{t}s_{2j}\left[\binom{t}{j}\right]_2 \right]_2 \sum_{s_1}\cdots\sum_{s_{2t-1}}\prod_{i=0}^{2t}a_0^{s_i}(1-a_0)^{1-s_i} \\
&= \sum_{s_0}\sum_{s_1}\cdots\sum_{s_t} \left[ \sum_{j=0}^{t}s_{j}\left[\binom{t}{j}\right]_2 \right]_2 \prod_{i=0}^{t}a_0^{s_{i}}(1-a_0)^{1-s_{i}}.
\end{split}
\label{eqn:spiffy}
\end{equation}

In Fig.\ \ref{fig:binary} we have plotted the macroscopic map $\psi_t(a_0)$ for $a_0 = 0.1$ and $t = 0,\ldots,16$ along with the observable $G_{n,t}$ for $n = 10,000$.  To generate the initial lattice of $n$ sites with magnetization $a_0$, the first $\floor{na_0}$ sites were assigned a value of 1, while the remaining sites were assigned 0.  The sites were then randomly relabeled to eliminate site-dependent correlations.  This gave the desired conditional distribution, and $G_{n,t}$ was determined by time-evolving this initial lattice.  Both $\psi_t(a_0)$ and $G_{n,t}$ are in good agreement, corroborating our assumption that the free energy is both well defined and differentiable.

It may be noted that the graph of $\psi_t(a_0)$ repeats its minimum value at each $t=2^k$, where $k=0,1,2,\ldots$.  This becomes clearer when we view the long-time behavior of $G_{n,t}$, as shown in Fig.\ \ref{fig:binary-long}.  In addition to this repetition of the minimum value, there are higher, more closely spaced bands of increasing density, with the greatest density appearing near $0.5$.  Larger values of $n$ give bands of narrower vertical width, suggesting that $\psi_t(a_0)$ attains values only within each band.  That this is indeed the case is shown below, and an explicit formula for $\psi_t(a_0)$ giving the explicit value of the $m^{\text{th}}$ band will be derived.  The final result is given in Eq.\ (\ref{eqn:nifty}) below.

The first (lowest) band may be understood by considering the formula for $\psi_t(a_0)$ given in Eq.\ (\ref{eqn:spiffy}).  Notice that for $t=2^k$ the binomial coefficient $\binom{t}{j} = \binom{2^k}{j}$ is an even number for all $0 < j < t$.  This means that the sum (modulo 2) over $s_j[\binom{t}{j}]_2$ can have nonzero contributions only from the first ($j=0$) and last ($j=t$) terms.  This leaves four possible values for the pair $s_0, s_t$, two of which give zero contributions.  This gives
\begin{equation*}
\begin{split}
\psi_t(a_0) &= \sum_{s_0}\cdots\sum_{s_t}\left[ s_0\left[\binom{t}{0}\right]_2 + s_t\left[\binom{t}{t}\right]_2 \right]_2 \prod_{i=0}^{t}a_0^{s_i}(1-a_0)^{1-s_i} \\
&= \sum_{s_0}\sum_{s_t} \left[ s_0 + s_t \right]_2 a_0^{s_0}(1-a_0)^{1-s_0} a_0^{s_t}(1-a_0)^{1-s_t} \\
&= 2a_0(1-a_0).
\end{split}
\end{equation*}
Comparison with Fig.\ \ref{fig:binary} shows that, for $a_0 = 0.1$, the first band does indeed occur at the value $2a_0(1-a_0) = 0.18$.

In a similar manner, values in the second band appear to occur at times $2^k+\set{1, 2, 4, \ldots, 2^{k-1}}$.  More generally, values on the $m^{\text{th}}$ band occur at times $t = 2^{k_1}+2^{k_2}+\cdots+2^{k_m}$, where $0 \le k_m < \ldots < k_1$.  Now, any $t$ may be written in this form, so $m$ is really just the sum of digits in the binary expansion of $t$.  Written in this way, we observe that $\binom{t}{j}$ will be odd only for the values
\begin{multline}
j,\; n-j = 0, 2^{k_2}, \set{2^{k_3}, 2^{k_2}+2^{k_3}}, \\
\set{2^{k_4},2^{k_2}+2^{k_4}, 2^{k_3}+2^{k_4},2^{k_2}+2^{k_3}+2^{k_4}}, \ldots
\end{multline}
Consequently, the sum $\sum_{j=0}^{t}s_{j}[\binom{t}{j}]_2$ will contain at most $2\times2^{m-1} = 2^m$ nonzero terms.  (This fact was also noted by Wolfram in connection to his study of cellular automata using rule 90.  See \cite{Wolfram1984} and references within.)  Reordering the summation indices, we may therefore write
\begin{equation*}
\psi_t(a_0) = \sum_{s_0}\cdots\sum_{s_{2^m-1}} \left[ \sum_{j=0}^{2^m-1}s_{j} \right]_2 \prod_{i=0}^{2^m-1}a_0^{s_{i}}(1-a_0)^{1-s_{i}}.
\end{equation*}
For each choice of $s_0,\ldots, s_{2^m-1}$, only an odd number of nonzero spins will contribute, since we take the innermost sum modulo 2.  There are $\binom{2^m}{2\nu-1}$ choices with exactly $2\nu-1$ nonzero spins, taking $\nu = 1,\ldots,2^{m-1}$, and for each such choice a value $a_0^{2\nu-1}(1-a_0)^{2^m-2\nu+1}$ is contributed to the sum.  Summing over all choices, we thereby obtain the surprisingly compact result
\begin{equation}
\psi_t(a_0) = \sum_{\nu=1}^{2^{m-1}} \binom{2^m}{2\nu-1} a_0^{2\nu-1}(1-a_0)^{2^m-2\nu+1},
\label{eqn:nifty}
\end{equation}
where $m$, as noted earlier, is the sum of digits in the binary expansion of $t$. (For $t=0$ the formula is not valid, but then $\psi_0$ is just the identity.)  Thus, we have obtained the exact macroscopic solution for an observable, here the ``magnetization,'' from the exact microscopic dynamics of an interacting system.

%%%%%%%%%%%%%%%%%%%%%%%%%%%%%%%%%%%%%%%%%%%%%%%%%%%%%%%%%%%%%%%%%%%%%%%%%%%%

\section{Discussion}
\label{sec:discussion}

In this work we have considered the prediction of a final macroscopic state from its given initial value in the case that the observable may be vector valued and the underlying microscopic dynamics may include interactions amongst the constituents.  The approach we have used borrows well-known techniques from large deviation theory to establish a level-1 LDP for the joint initial and final macrostates by examining regularity conditions of a generalized, dynamic free energy function.  From this, a conditional LDP was deduced which gave the desired convergence in probability as well as an explicit expression for the predicted final macrostate in terms of a canonical expectation.

The dynamic free energy was defined in terms of a given macroscopic observable with an associated scaling parameter, the underlying microscopic dynamics of the system, and an \textit{a priori} probability distribution on the space of microstates.  A key assumption made was that the dynamic free energy is everywhere well defined, finite, and differentiable.  Additionally, we assumed a one-to-one correspondence between macrostates and the conjugate macrostates arising from the Legendre transform of the free energy.  No claim has been made that these properties should hold universally, though we submit that they are reasonable.

One of the goals of this endeavor was to treat macroscopic dynamics on an equal footing with dynamics on the microscopic level.  To this end, we have considered several properties of the macroscopic dynamics, as given by the set of predicted macrostates, which may be inherited from the underlying microscopic description.  We have found that in general the group and semigroup properties are not inherited, though semigroups are known to be possible.  Time reversal invariance at the microscopic level was found to imply time symmetry at the macroscopic level if the observable and \textit{a priori} measure bear the proper symmetries.  Equilibration, the tendency for the predicted macrostate to tend to a limiting value in the long-time limit, was found to hold under mixing dynamics; however, the rate of convergence must not depend greatly upon the size of the system.  Finally, linear or, more generally, affine transformations of the macroscopic variables where found to preserve macroscopic determinism, with the macroscopic map undergoing a corresponding covariant transformation.

Results for noninteracting systems, treated previously by different techniques, were rederived and extended to include vector-valued macrostates.  This case also provided an example in which all regularity conditions could be verified rigorously and completely.  Interacting systems are, of course, more difficult to analyze.  We therefore considered a simple example consisting of a binary-state lattice with nearest-neighbor interactions.  An explicit expression for the macroscopic map was derived, whose predictions agreed well with the results of numerical simulations, but the regularity conditions for the dynamic free energy could only be verified numerically.  Due to its analytic intractability, there remains a need for methods of verifying the regularity conditions of the dynamic free energy without explicit computation.

% begin CHANGE (2nd revision)
%%%%%%%%%%%%%%%%%%%%%%%%%%%%%%%%%%%%%%%%%%%%%%%%%%%%%%%%%%%%%%%%%%%%%%%%%%%%
\appendix
%%%%%%%%%%%%%%%%%%%%%%%%%%%%%%%%%%%%%%%%%%%%%%%%%%%%%%%%%%%%%%%%%%%%%%%%%%%%

\section{Proof of Conditional LDP}
\label{app:Conditional_LDP}

To prove Theorem \ref{thm:Conditional_LDP}, we shall first require the following results.

%------------------------------------------------------------------------------

\begin{lemma}
\label{lem:continuity}
If $I$ is a good, finite, continuous rate function, then $\inf I(A) = \inf I(\closure{A})$ for any $A$.
\end{lemma}

\begin{proof}
If $A$ is empty the result holds trivially, so we shall assume henceforth that $A$ is nonempty.  As $I$ is everywhere finite, this implies that $\inf I(A) < \infty$.  Since $I$ is a good rate function, there exists at least one $x_A \in \closure{A}$ such that $I(x_A) = \inf I(\closure{A})$.  Furthermore, since $A$ is a nonempty subset of $\closure{A}$, we must have $I(\closure{A}) \le \inf I(A)$.  Thus, $I(x_A) = \inf I(\closure{A}) \le \inf I(A) < \infty$.  We shall now assume only a strict inequality holds, i.e., $\inf I(\closure{A}) < \inf I(A)$, and show that this entails a contradiction.

Assume $\inf I(\closure{A}) < \inf I(A)$ and let $V = (-\infty, \inf I(A))$.  Since $I$ is continuous and $I(x_A) \in V$ by assumption, there exists a neighborhood $U$ of $x_A$ such that $I(U) \subseteq V$.  Since $x_A \in \closure{A}$ and $U$ is a neighborhood of $x_A$, there exists an $x_A' \in U \cap A$.  Since $x_A' \in U$, $I(x_A') \in V$ and hence $I(x_A') < \inf I(A)$.  However, since $x_A' \in A$ as well, $I(x_A') \ge \inf I(A)$.
\end{proof}

%------------------------------------------------------------------------------

\begin{lemma}
\label{lem:topological}
For any two sets $A$ and $B$, $\interior{A} \cap \closure{B} \subseteq \closure{\interior{A} \cap B}$.
\end{lemma}

\begin{proof}
If $\interior{A} \cap \closure{B} = \varnothing$ then we are done, so suppose there exists an $x \in \interior{A} \cap \closure{B}$.  We will have $x \in \closure{\interior{A} \cap B}$ iff for any neighborhood $U$ of $x$ we have that $U \cap (\interior{A} \cap B) \neq \varnothing$.  Now, given $U$, $U \cap \interior{A}$ is also a neighborhood of $x$, so, since $x \in \closure{B}$, it follows that $(U \cap \interior{A}) \cap B \neq \varnothing$.
\end{proof}

%------------------------------------------------------------------------------

\begin{corollary}
\label{cor:final}
If $I$ satisfies Lemma \ref{lem:continuity}, then $\inf I(\interior{A} \cap \interior{B}) = \inf I(\interior{A} \cap \closure{B})$ for any two sets $A$ and $B$.
\end{corollary}

\begin{proof}
By Lemmas \ref{lem:continuity} and \ref{lem:topological},
\begin{equation*}
\inf I(\interior{A} \cap \interior{B}) \ge \inf I(\interior{A} \cap \closure{\interior{B}}) \ge \inf I(\closure{\interior{A} \cap \interior{B}}) = \inf I(\interior{A} \cap \interior{B}).
\end{equation*}
\end{proof}

%------------------------------------------------------------------------------
\bigskip
\begin{proof}{(Conditional LDP)}
Observe that the large deviation upper and lower bounds imply that, for any $\varepsilon > 0$ and all $n$ sufficiently large,
\begin{subequations}
\begin{align}
v_n^{-1}\log P_n(A) &< -(1-\varepsilon)\inf I(\closure{A})\:\: \quad\text{for $0 < \inf I(\closure{A}) < \infty$,} \\
v_n^{-1}\log P_n(A) &> -(1+\varepsilon)\inf I(\interior{A})    \quad\text{for $0 < \inf I(\interior{A}) < \infty$.}
\end{align}
\end{subequations}
Similarly, $\inf I(\interior{A}) = 0$ implies $v_n^{-1}\log P_n(A) > -\varepsilon$ for all $n$ sufficiently large.  As $I$ is assumed finite, $\inf I(A) = \infty$ implies $A = \varnothing$.

By Lemma \ref{lem:continuity} $\inf I(\interior{B}) = \inf I(\closure{\interior{B}})$, and, by assumption $\closure{\interior{B}} = \closure{B}$, so $\inf I(\interior{B}) = \inf I(B) = \inf I(\closure{B})$.  As $B$ is assumed nonempty, $\inf I(B) < \infty$, so the large deviation lower bound implies $P_n(B) > 0$ for $n$ sufficiently large.

We begin with the large deviation upper bound.  First assume $0 < \inf I(\closure{A \cap B}) < \infty$ and $\inf I(B) > 0$.  For a given $\varepsilon > 0$ we have that for all $n$ sufficiently large
\begin{align*}
v_n^{-1} \log P_n(A|B)
&= v_n^{-1} \log P_n(A\cap B) - v_n^{-1} \log P_n(B) \\
&< -(1-\varepsilon) \inf I(\closure{A \cap B}) + (1+\varepsilon) \inf I(\interior{B}) \\
&\le -\left[\inf I(\closure{A}\cap\closure{B}) - \inf I(B)\right] + \varepsilon' \\
&= -\inf I_{B}(\closure{A}) + \varepsilon',
\end{align*}
where $\varepsilon' =  \varepsilon \left[\inf I(\closure{A\cap B}) + \inf I(B)\right]$.  As $\varepsilon' > 0$ and may be made arbitrarily small, we conclude that
\begin{equation*}
\limsup_{n\to\infty} \;v_n^{-1} \log P_n(A|B) \le -\inf I_B(\closure{A}).
\end{equation*}
If $0 < \inf I(\closure{A \cap B}) < \infty$ yet $\inf I(B) = 0$, then
\begin{align*}
v_n^{-1} \log P_n(A|B) &< -(1-\varepsilon) \inf I(\closure{A \cap B}) + \varepsilon \\
&\le -\inf I_{B}(\closure{A}) + \varepsilon \left[\inf I(\closure{A\cap B}) + 1\right]
\end{align*}
and the upper bound is again found to hold.

If $\inf I(\closure{A\cap B}) = 0$, then $\inf I(B) = \inf I(\closure{B}) \le \inf I(\closure{A}\cap\closure{B}) \le \inf I(\closure{A \cap B}) = 0$ and $\inf I_B(\closure{A}) = \inf I(\closure{A}\cap\closure{B}) - \inf I(B) \le \inf I(\closure{A \cap B}) - \inf I(B) = 0$.  Since $v_n^{-1} \log P_n(A|B) \le 0 = -\inf I_B(\closure{A})$, the upper bound is clearly satisfied.

If $\inf I(\closure{A\cap B}) = \infty$, then $P_n(A|B) = 0$ and $v_n^{-1} \log P_n(A|B) = -\infty$ for all $n$ sufficiently large.  Thus, $\limsup_{n\to\infty} v_n^{-1} \log P_n(A|B) = -\infty \le -\inf I_B(\closure{A})$.

For the large deviation lower bound, suppose $0 < \inf I(\interior{(A\cap B)}) < \infty$ and note that for all $n$ sufficiently large,
\begin{align*}
v_n^{-1} \log P_n(A|B)
&> -(1+\varepsilon)\inf I(\interior{(A\cap B)}) + (1-\varepsilon)\inf I(\closure{B}) \\
&= -\left[\inf I(\interior{A}\cap\interior{B}) - \inf I(B)\right] - \varepsilon' \\
&= -\inf I_B(\interior{A}) - \varepsilon',
\end{align*}
where $\varepsilon' =  \varepsilon\left[\inf I(\interior{A}\cap\interior{B}) + \inf I(B)\right]$ and Corollary \ref{cor:final} has been used in the last equality.  The second term, $\varepsilon'$ is positive and may be made arbitrarily small, so we conclude
\begin{equation*}
\liminf_{n\to\infty} v_n^{-1} \log P_n(A|B) \ge -\inf I_B(\interior{A}).
\end{equation*}

If $\inf I(\interior{A}\cap\interior{B}) = 0$, then $\inf I(B) = \inf
I(\closure{B}) \le \inf I(\interior{A}\cap\closure{B}) = \inf
I(\interior{A}\cap\interior{B}) = 0$ and $\inf I_B(\interior{A}) =
\inf I(\interior{A}\cap\closure{B}) - \inf I(B) = \inf
I(\interior{A}\cap\interior{B}) - \inf I(B) = 0$.  However, for any
given $\varepsilon > 0$ and all $n$ sufficiently large,
\begin{equation*}
v_n^{-1} \log P_n(A|B) > -\varepsilon + (1-\varepsilon) \inf I(\closure{B}) = -\varepsilon = -\inf I_B(\interior{A}) - \varepsilon.
\end{equation*}
Thus, the lower bound is satisfied in this case.

Finally, if $\inf I(\interior{A}\cap\interior{B}) = \infty$, then $\inf I_B(\interior{A}) = \inf I(\interior{A}\cap\closure{B}) - \inf I(B) = \inf I(\interior{A}\cap\interior{B}) - \inf I(B) = \infty$.  Since $v_n^{-1} \log P_n(A|B) \ge -\infty = -\inf I_B(\interior{A})$, the lower bound is clearly satisfied in this case as well.

To complete the proof, we must show that $I_{B}$ is a good, continuous rate function relative to $\closure{B}$.  Effective continuity of $I_B$ follows from that of $I$.  To show that $I_{B}$ is a good rate function, consider any $\alpha < \infty$ and note that
\begin{align*}
\set{x\in X: I_B(x) \le \alpha} 
&= \set{x\in \closure{B}: I(x)-\inf I(B) \le \alpha} \\
&= \set{x\in X: I(x) \le \alpha + \inf I(B)} \cap \closure{B}.
\end{align*}
We have already established that $\inf I(B) < \infty$.  As $X$ is a Hausdorff space and $I$ is a good rate function, the above intersection is compact, thus establishing that $I_B$ is a good rate function.
\end{proof}

%%%%%%%%%%%%%%%%%%%%%%%%%%%%%%%%%%%%%%%%%%%%%%%%%%%%%%%%%%%%%%%%%%%%%%%%%%%%

\section{Determination of \protect$\inf I_t(B_0)$}
\label{app:Conditional_Infimum}

For any $(a_1,a_2) \in M \times M$ we have
\begin{equation*}
\begin{split}
I_t(a_1,a_2) &= \lambda_1 a_1 + \lambda_2 a_2 - \Psi_t(\lambda_1,\lambda_2) \\
             &= \lambda_1 a_1 + \lambda_2 a_2 - \Psi_t(\lambda_1,\lambda_2) + I_t(a_0,a_t) - [\lambda_0 a_0 - \Psi_t(\lambda_0,0)],
\end{split}
\end{equation*}
where $a_t = \nabla_2 \Psi_t(\lambda_0,0)$.  If $(a_1,a_2) \in B_0$, then $\lambda_0 a_0 \le \lambda_0 a_1$, so
\begin{equation*}
I_t(a_1,a_2) \ge [\lambda_1 a_1 + \lambda_2 a_2 - \Psi_t(\lambda_1,\lambda_2) - \lambda_0 a_1 + \Psi_t(\lambda_0,0)] + I_t(a_0,a_t).
\end{equation*}
The term is brackets is nonnegative, since
\begin{equation*}
\begin{split}
\lambda_1 a_1 + \lambda_2 a_2 - \Psi_t(\lambda_1,\lambda_2) &= \sup_{\lambda_1',\lambda_2'} [\lambda_1'a_1 + \lambda_2'a_2 - \Psi_t(\lambda_1',\lambda_2)'] \\
&\ge \lambda_0 a_1 + 0\,a_2 - \Psi(\lambda_0,0)
\end{split}
\end{equation*}
with equality iff $(\lambda_1,\lambda_2) = (\lambda_0,0)$.  Thus, $I_t(a_1,a_2) \ge I_t(a_0,a_t)$ for all $(a_1,a_2) \in B_0$.  Since $(a_0,a_t)$ is itself in $B_0$, we conclude that
\begin{equation*}
\inf I_t(B_0) = I_t(a_0,a_t) = \lambda_0 a_0 - \Psi_t(\lambda_0,0),
\end{equation*}
where $a_t = \nabla_2\Psi_t(\lambda_0,0)$.

% end CHANGE
%%%%%%%%%%%%%%%%%%%%%%%%%%%%%%%%%%%%%%%%%%%%%%%%%%%%%%%%%%%%%%%%%%%%%%%%%%%%

\section*{Acknowledgments}

This work was supported in part by the Engineering Research Program of the Office of Basic Energy Sciences at the U.S.\ Department of Energy, Grant No.\ DE-FG0394ER14465.
One of us (B.L.) has also received partial funding by the Applied Research Laboratories of the University of Texas at Austin, Independent Research and Development Grant No.\ 926.
Finally, the authors would like to thank the referees for their valuable comments and suggestions.

%%%%%%%%%%%%%%%%%%%%%%%%%%%%%%%%%%%%%%%%%%%%%%%%%%%%%%%%%%%%%%%%%%%%%%%%%%%%

% \bibliographystyle{unsrt}
% \bibliography{md2-rev2}

\newpage
%%%%%%%%%%%%%%%%%%%%%%%%%%%%%%%%%%%%%%%%%%%%%%%%%%%%%%%%%%%%%%%%%%%%%%%%%%%%
% Figure Captions
%%%%%%%%%%%%%%%%%%%%%%%%%%%%%%%%%%%%%%%%%%%%%%%%%%%%%%%%%%%%%%%%%%%%%%%%%%%%
\pagestyle{empty}

\begin{figure}[p]
\caption{\setlength{\baselineskip}{2em}Contour plot of the joint probability density for $(G_{n,0},G_{n,t})$.  The \textit{a priori} most likely macrostate is $a_*$, while $a_0$ is the given initial macrostate.  The conditioning set corresponding to $a_0$ is $B_0$, and $a_t$ is the most likely value of $G_{n,t}$ under this conditioning.}
\label{fig:contours}
\end{figure}

\begin{figure}[p]
\caption{\setlength{\baselineskip}{2em}Plot of the finite-$n$ dynamic free energy, $\Psi_{n,t}(\lambda_1,\lambda_2)$ for $t=1$ and $n=10$.  The observable is the mean spin or ``magnetization'' and the dynamics is that of the sum modulo 2 rule.}
\label{fig:binary-free}
\end{figure}

\begin{figure}[p]
\caption{\setlength{\baselineskip}{2em}Plot of the expected ($\circ$) and actual ($+$) magnetization for a binary lattice evolving according to the sum modulo 2 rule.  The size of the lattice is $n=10,000$.}
\label{fig:binary}
\end{figure}

\begin{figure}[p]
\caption{\setlength{\baselineskip}{2em}Plot of the magnetization, $G_{n,t}$, versus time for a binary lattice.  The discrete bands are indicated by solid horizontal lines with the lowest band being the first.  Bands above the fifth are too closely spaced to be resolved.}
\label{fig:binary-long}
\end{figure}

\newpage
%%%%%%%%%%%%%%%%%%%%%%%%%%%%%%%%%%%%%%%%%%%%%%%%%%%%%%%%%%%%%%%%%%%%%%%%%%%%
% Figures
%%%%%%%%%%%%%%%%%%%%%%%%%%%%%%%%%%%%%%%%%%%%%%%%%%%%%%%%%%%%%%%%%%%%%%%%%%%%

\begin{figure}[ht]
\centerline{\scalebox{0.8}{\includegraphics{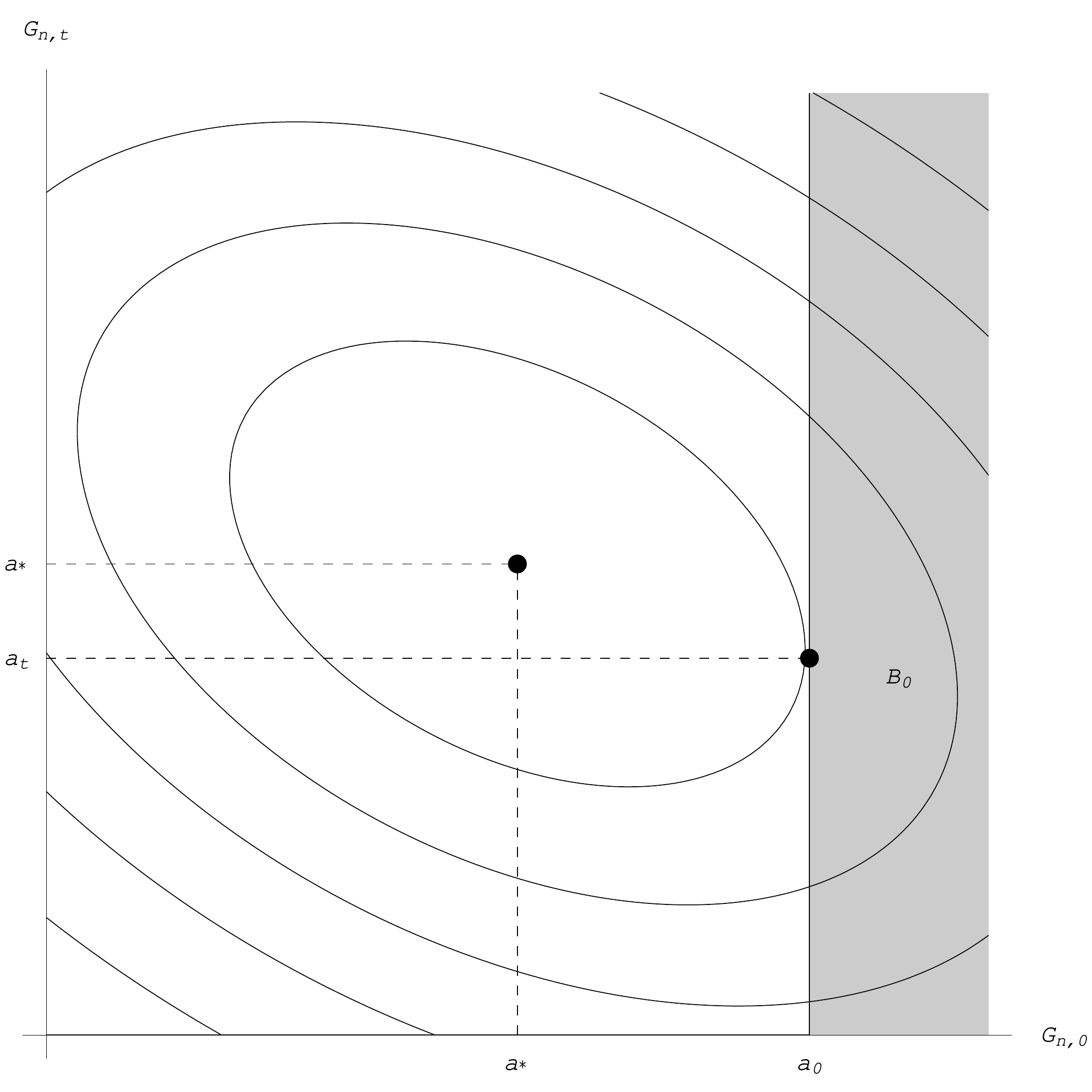}}}
\centerline{\Large Figure \ref{fig:contours}}
\end{figure}

\begin{figure}[ht]
\centerline{\scalebox{2}{\includegraphics{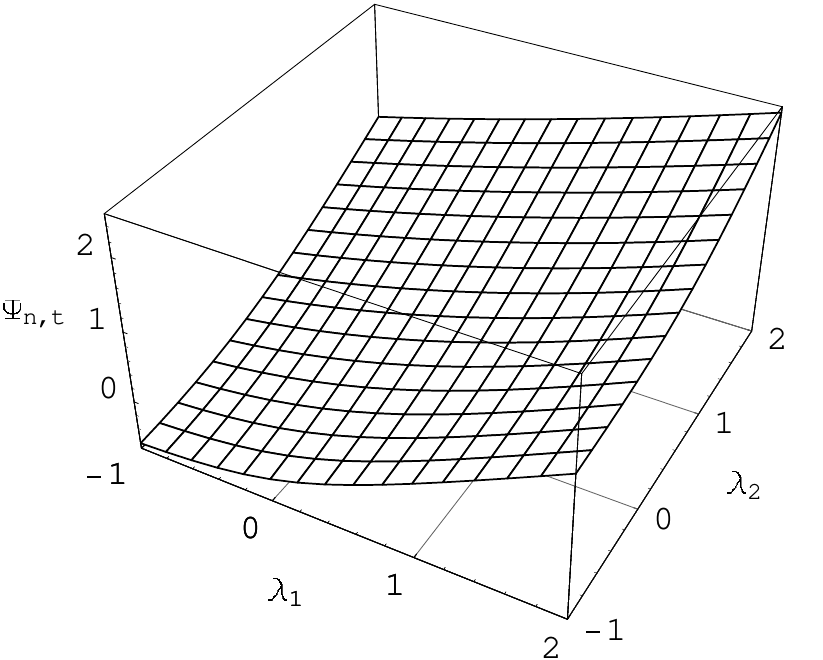}}}
\centerline{\Large Figure \ref{fig:binary-free}}
\end{figure}

\begin{figure}[ht]
\centerline{\scalebox{.65}{\includegraphics{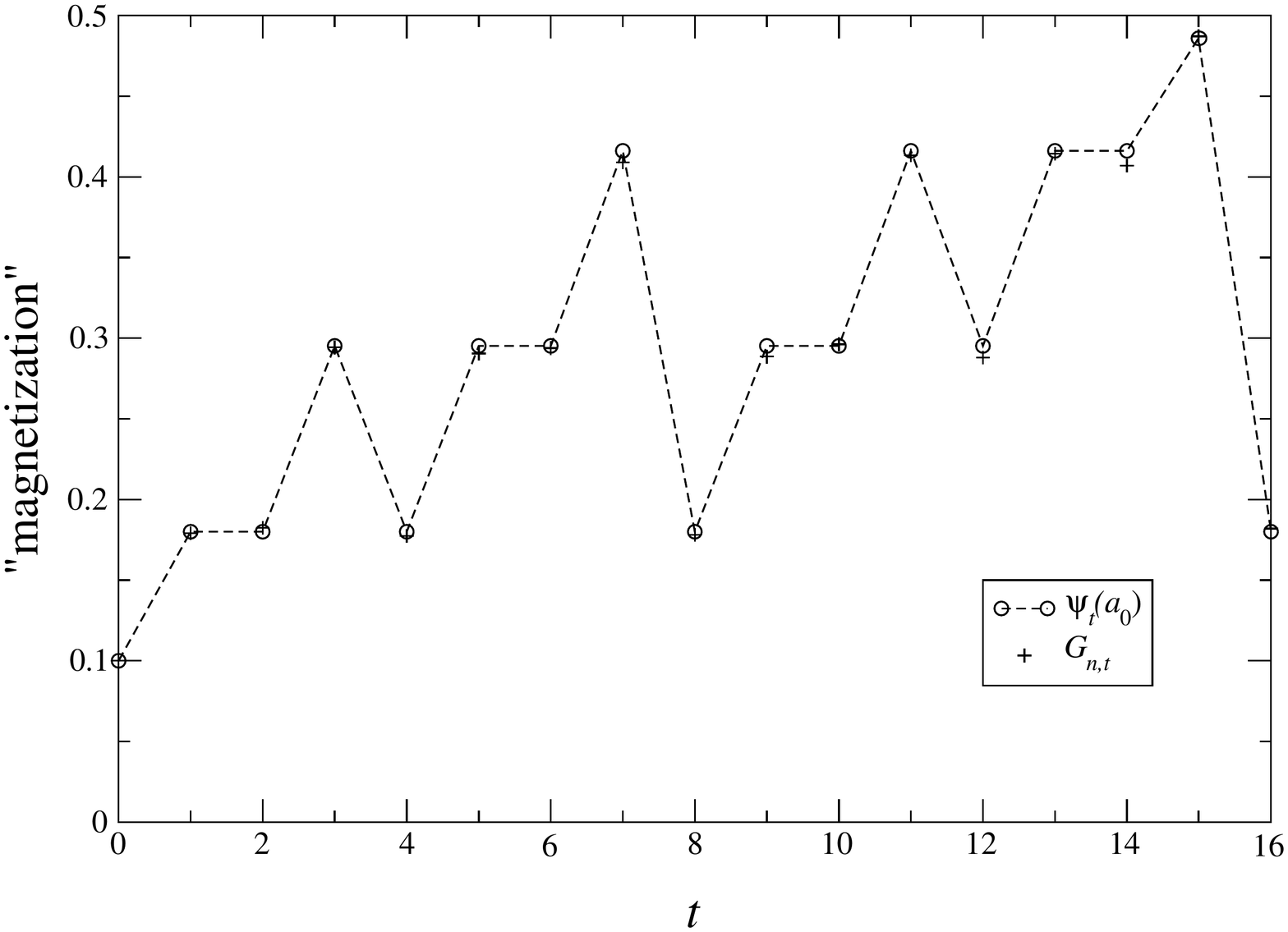}}}
\centerline{\Large Figure \ref{fig:binary}}
\end{figure}

\begin{figure}[hp]
\centerline{\scalebox{.65}{\includegraphics{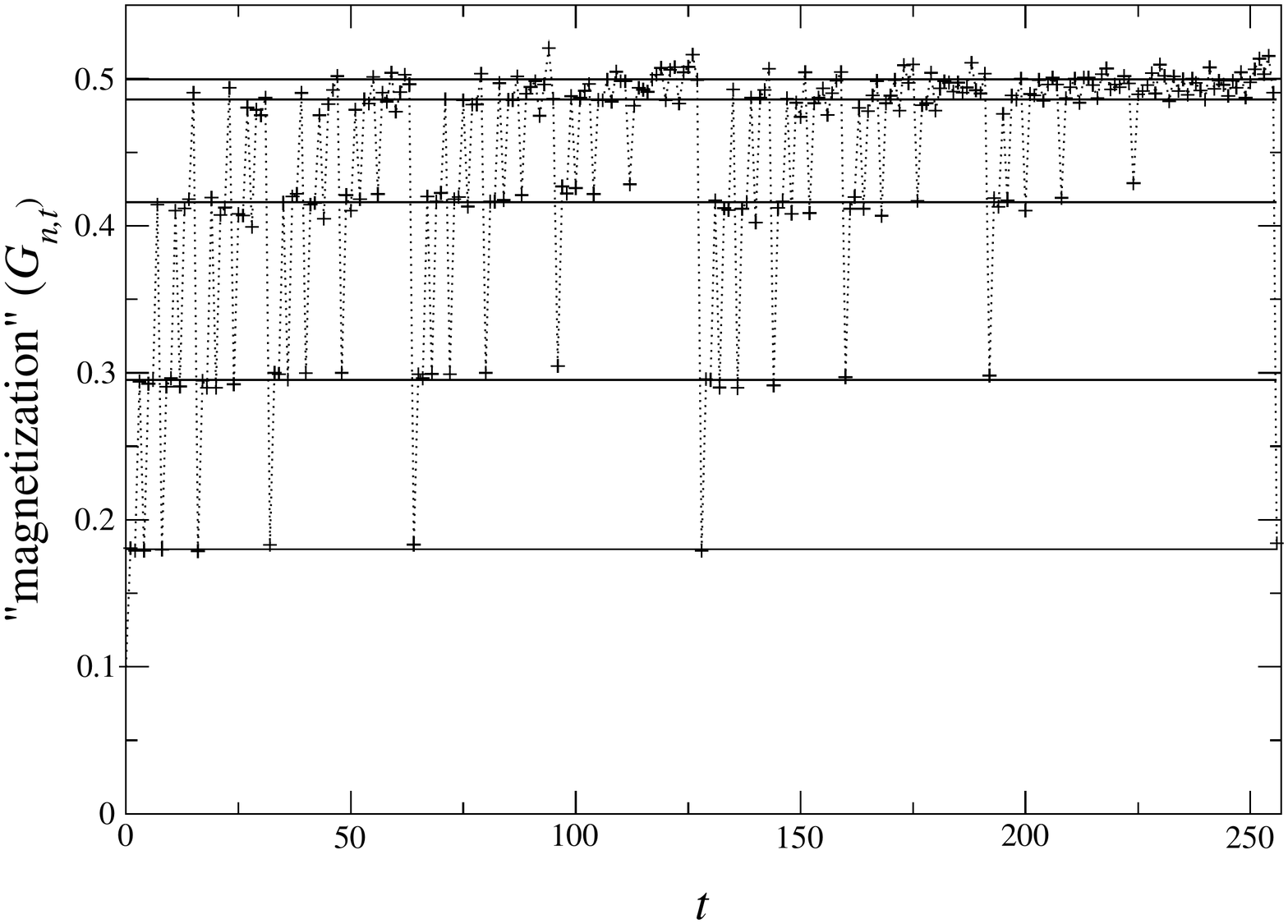}}}
\centerline{\Large Figure \ref{fig:binary-long}}
\end{figure}

%%%%%%%%%%%%%%%%%%%%%%%%%%%%%%%%%%%%%%%%%%%%%%%%%%%%%%%%%%%%%%%%%%%%%%%%%%%%

\end{document}